\documentclass[journal]{IEEEtran}

\usepackage[OT1]{fontenc} 
\usepackage{comment}
\usepackage{braket}
\usepackage{pict2e}
\usepackage{svg}
\usepackage{amsmath}
\usepackage{amssymb}
\usepackage{amsthm}
\usepackage{dsfont}
\usepackage{xcolor}
\usepackage{graphicx}
\usepackage{mathtools}
\usepackage{accents}
\usepackage{mdframed}
\usepackage{tikz}

\graphicspath{{grafics/}}

\newcommand{\R}[1]{\mathfrak{#1}}

\newcommand{\CP}[2]{\mathbb{P} \mathopen{} \left[ #1 \, \middle| \, #2 \right] \mathclose{} }
\newcommand{\E}[1]{\mathbb{E} \mathopen{} \left[ #1 \right] \mathclose{} }

\newcommand{\CE}[2]{\mathbb{E} \mathopen{} \left[ #1 \, \middle| \, #2 \right] \mathclose{} }

\newcommand{\tr}[1]{\tilde{#1}}
\newcommand{\e}[1]{\bar{#1}}

\newcommand{\conv}[1]{ \operatorname{conv} \mathopen{} \left( #1 \right) \mathclose{} }

\newcommand{\MW}{\mathsf{MW}}
\newcommand{\PNC}{\mathsf{PNC}}
\newcommand{\fPNC}{\mathsf{fPNC}}

\newcommand{\Rminus}{R^-}

\newcommand{\D}{\Delta}

\newcommand{\figref}[1]{\figurename~\ref{#1}}

\newcommand{\norm}[1]{\left\lVert#1\right\rVert}

\allowdisplaybreaks

\DeclareMathOperator*{\argmin}{\arg\min}
\DeclareFontFamily{U}{mathb}{}
\DeclareFontShape{U}{mathb}{m}{n}{
  <-5.5> mathb5
  <5.5-6.5> mathb6
  <6.5-7.5> mathb7
  <7.5-8.5> mathb8
  <8.5-9.5> mathb9
  <9.5-11.5> mathb10
  <11.5-> mathbb12
}{}

\theoremstyle{plain}
\newtheorem{theorem}{Theorem}

\newtheorem{lemma}[theorem]{Lemma}

\begin{document}

\title{Model-Predictive Control for Discrete-Time Queueing Networks with Varying Topology}

\author{Richard~Schoeffauer\textsuperscript{1}
        and~Gerhard~Wunder\textsuperscript{2}
\thanks{\textsuperscript{1} richard.schoeffauer@fu-berlin.de, \textsuperscript{2} gerhard.wunder@fu-berlin.de}
\thanks{Both authors are members of the Heisenberg Communication and Information Theory Group at Freie Universität Berlin}
}

\maketitle
\begin{abstract}
In this paper, we equip the conventional discrete-time queueing network with a Markovian input process, that, in addition to the usual short-term stochastics, governs the mid- to long-term behavior of the links between the network nodes. This is reminiscent of so-called \textit{Jump-Markov systems} in control theory
and allows the network topology to change over time.
We argue that the common back-pressure control policy is inadequate to control such network dynamics and propose a novel control policy inspired by the paradigms of \textit{model-predictive control}.
Specifically, by defining a suitable but arbitrary prediction horizon, our policy takes into account the future network states and possible control actions.
This stands in clear contrast to most other policies which are myopic, i.e. only consider the next state.
We show numerically that such an approach can significantly improve the control performance and introduce several variants, thereby trading off performance versus computational complexity.
In addition, we prove so-called \textit{throughput optimality} of our policy which guarantees stability for all network flows that can be maintained by the network.
Interestingly, in contrast to general stability proofs in model-predictive control, our proof does not require the assumption of a terminal set (i.e. for the prediction horizon to be large enough). Finally, we provide several illustrating examples, one of which being a network of synchronized queues. This one in particular constitutes an interesting system class, in which our policy exerts superiority over general back-pressure policies, that even lose their throughput optimality in those networks.
\end{abstract}

\begin{IEEEkeywords}
Predictive Network Control, Model Predictive Control, Jump-Markov Systems, Throughput Optimality
\end{IEEEkeywords}

\IEEEpeerreviewmaketitle

\section{Introduction and related Research}

Discrete-time queueing networks are used to model a variety of scenarios, ranging from traffic control over parallel computing to wireless communication. They are closely related to the canonical control system
\begin{equation}
	\label{eq::usual_control}
	x_{t+1} = A x_t + B_t v_t + D_t w_t
\end{equation}
with some significant differences: i) The controls $v_t$ are binary in nature and linearly constrained by $C v_t \leq c$, e.g. due to the interference properties of wireless channels. ii) The state lives on the discrete set $x_t \in \mathbb{N}^n$ where it exhibits no inertia ($A = I$). iii) And crucially, the matrices $B_t$ and $D_t$ behave stochastically, implying that the effect of a control decision is not certain. Together with the class of back-pressure control policies, those systems form a well investigated subclass of control problems.

The prototype back-pressure policy, that we will call the max-weight policy ($\MW$), was first introduced in \cite{Tassiulas1992}, where the authors also proved its much praised property of being \textit{throughput optimal}. This means, that $\MW$ can manage any load of traffic, provided this load can somehow be supported by the network topology.
Over time, many variations of $\MW$ where developed, e.g. to allow for a generalized control objective \cite{Meyn2009} \cite{Kasparick2018}, or to increase its performance in special cases like networks with input-queued switches or time-varying channels \cite{Mckeown1999} \cite{Neely2005}. Specific shortcomings of $\MW$, like e.g. high end-to-end delay, where investigated in \cite{Khan2009} \cite{Subramanian2007} \cite{Ying2011} and later partially remedied by \cite{Neely2005} \cite{Ying2011} \cite{Huang2013} \cite{Xiong2011}, using e.g. shortest path algorithms to reduce delay especially in low traffic scenarios.

In this paper, we propose a novel control policy that is predictive in nature and that we will call predictive network control ($\PNC$). It can be regarded as a generalization of $\MW$, since it contains $\MW$ as a special case. But while $\MW$ and all its derivations are \textit{myopic}, i.e. only aim to improve the system state for the \textit{immediate} next time slot, $\PNC$ aims to improve the system state for \textit{multiple} time slots up until a prediction horizon. This leads to the calculation of an entire optimal \textit{trajectory} of control vectors. However, instead of applying the entire trajectory for the next few time slots, only the first control vector is applied to the system and the process repeats in the consecutive time slot. This allows the controller to react to any unforeseen changes in the control system \cite{Mayne2000}.

Such a control scheme is called model-predictive control (MPC), and therefore $\PNC$ is a realization of MPC, tailored specifically towards queueing networks. MPC itself is a well established branch of control theory and can cope very easily with hard constraints and non-linearities, making it particularly suited for our control problem. However, its advantages are payed for by high requirements on computational resources.
So far, there has only been one attempt to bring MPC to queueing networks: In \cite{VanLeeuwaarden2010a} the authors focus on a special case of the standard model, in which only the arrivals to the system are of stochastic nature. The investigation is limited to numerical simulations, which show better system performance (smoother time behavior) for a designed MPC controller compared to simple feedback control laws.
Since our queueing network will include a much higher degree of stochastics, we will not follow up on their work.

Because a queueing network misses any inertia ($A = I$), a predictive control scheme can only tap into its full potential, if the stochastics for $B_t$ (or $D_t$) are complex enough. E.g. if both matrices behave according to white-noise, prediction over more than the next time slot yields close to no improvement over myopic strategies. Hence, the benefit of a predictive control scheme usually increases with complexity of the system model. Therefore, we let $B_t$ (the matrix which is responsible for the topology and the quality of the links between the nodes of the network) be governed by a discrete-time Markov chain (DTMC) and a Bernoulli trial. This gives the opportunity to model long-term and short-term effects, respectively. Take e.g. wireless relay networks with user mobility: here, short-term interference leading to packet loss can be modeled by the Bernoulli trial, while long-term change in channel quality due to the mobility can be expressed by the DTMC \cite{Guzman2019}.

Control systems, in which the model parameters change according to a DTMC are called Jump-Markov systems (JMS). (Since simple feedback controllers cannot detect this change, JMS are usually controlled with MPC controllers). There exist several control approaches towards JMS, covering cases with linear \cite{Park2002} \cite{Chitraganti2014} \cite{Tonne2017a} and even nonlinear system dynamics \cite{Liu2015} \cite{Tonne2017}, where the referenced works mainly differ in the choice of considered constraints.
However, all these works deal with conventional control systems, where the controller usually tries to follow a reference trajectory and noise ($w_t$) represents a stochastic disturbance with zero first moment. In contrast, from the perspective of queueing networks, the noise term represents the arrival of packets/customers whose first moment is strictly positive, and the controller tries to maintain finite queues for any given arrival (hence, there is no need for a reference trajectory). For that reason, prior work on JMS is only partially applicable to our systems. To the best of our knowledge, we are the first to consider both JMS and MPC in the context of discrete-time queueing networks.

Our \textbf{contribution} is three-folded: i) We develop a JMS-adapted discrete-time queueing network and introduce a family of predictive control policies, based on the paradigms of MPC. ii) We proof throughput optimality (the equivalent of stability) for the most simple of our predictive control policies, thereby implying the same for the rest. And iii) we show the benefit of these policies over the conventional back-pressure control ($\MW$), using numerical simulation. In particular, our policies seem to maintain their throughput optimality in networks with synchronized queues, making them unique.

\section{System Model \& Prerequisites}

\subsection{System Model}

We begin by stating the constituting equation for our system model and clarify its components afterwards. Similar to the conventional control system, a discrete-time queueing network can be expressed by its one-step evolution and associated constraints
\def\hspa{\hspace{0.5mm}}
\begin{gather}
    \label{eq::basic_queueing}
    q_{t+1} = q_t + R M_t v_t + a_t
    \\[1ex]
    \nonumber
    \text{subject to}
    \\[1ex]
    \nonumber
    \left(
	\begin{aligned}
	    C v_t & \leq c \\
	    - \Rminus v_t & \leq q_t
    \end{aligned}
    \right)
    \hspa \text{and} \hspa
    \left(
    \begin{gathered}
	    M_t \sim \mathcal{B}(W^{s_t})
	    \\
	    W^{s_t} \in \{W^1 ,\dots W^{n_s} \}
	    \\
	    (s_t) \sim \operatorname{DTMC}(\{1,\dots n_s\},P,s_0)
    \end{gathered}
    \right)
\end{gather}
The \textit{queue vector} $q_t = \left( q_t^1 \dots q^{n_q}_t \right)^\intercal \in \mathcal{Q} = \mathbb{N}^{n_q}$ represents the system- (or queue-) state, where $q_t^i$ counts the number of packets, waiting in queue $i = 1,\dots n_q$ in time slot $t$. Each queue itself is a node of the network.

In any time slot, packets can be transmitted from one queue to another if there exists a directed link between the two and the link is activated.
There are $n_v$ links, each of which can be represented by a vector $r^j \in \left\{ -1,0,+1 \right\}^{n_q}$ ($j = 1,\dots n_v$), that, by superpositioning with $q_t$, transfers a packet from one queue ($\{-1\}$) to another ($\{+1\}$). All links are collected as columns in the routing matrix $R \in \left\{ -1,0,+1 \right\}^{n_q \times n_v}$ which therefore holds the topology.

[Remark:
In \textit{conventional} networks, a link has exactly \textit{one} $\{-1\}$ entry (origin) and \textit{at most one} $\{+1\}$ entry (destination). This implicit constraint is a prerequisite for all back-pressure policies to develop their throughput optimality. Though we will also use this constraint throughout the paper, our novel control policies seem to maintain their throughput optimality even when it is violated (see section \ref{subsec::synchronized_queues}), allowing us to control networks with synchronized queues.]

The controller may activate a link in a given time slot via the binary control vector $v_t \in \{ 0 , 1 \}^{n_v}$.
If we could activate all links simultaneously ($v_t = \mathbf{1}_{n_v}$), the control problem would become trivial. However, we are usually constrained in the activation (e.g. due to interference properties) by the \textit{constituency constraint} $Cv_t \leq c$. The dimensions of $C$ and $c$ are case dependent, their entries are from the set $\mathbb{N}$. Furthermore, a packet can only be scheduled for transmission, if it is present at the corresponding queue, hence a packet may only traverse a single link per time slot. We will refer to this as the \textit{positiveness constraint}, which is readily implemented by considering the maximum one-step efflux of the system, which is $\Rminus v_t$, where $\Rminus$ is equal to $R$ without its positive entries. Naturally, the maximum efflux cannot drain more packets than are actually present: $q_t + \Rminus v_t \geq 0$. Note that this also guarantees that $q_t \in \mathbb{N}^{n_q}$.

For clarification, we refer to \figref{fig::min_example}. Here, we stated topology and constituency matrices and derived the corresponding constraints. 
Given only $C$ and $c$, both components of $v_t$ could be active simultaneously. However, if $q^2$ is empty ($q^2 = 0$), it is not possible to activate the second link $r^2$.

\def\comR{
	$
	R = \begin{pmatrix}
	-1 & 0 \\ +1 & -1
	\end{pmatrix}
	$
}
\def\comC{
	$
	C = \begin{pmatrix}
	0 & 0
	\end{pmatrix}
	$
}
\def\comc{
	$ c = 1 $
}
\def\comRC{
	$ \begin{pmatrix}
	1 & 0 \\ 0 & 1 
	\end{pmatrix} v_t \leq q_t$
}
\def\comCC{
	$ \begin{pmatrix}
		0 & 0 
	\end{pmatrix} v_t \leq 1$
}
\def\comA{
	$\Downarrow$
}

\begin{figure}[htbp]
	\centering
	\includegraphics[]{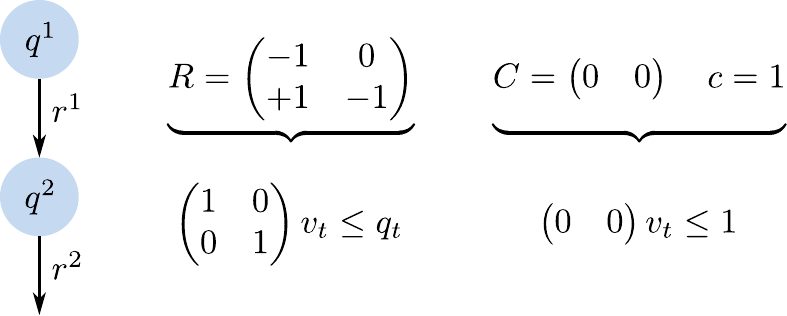}
	\caption{Minimal Example of a Queueing Network}
	\label{fig::min_example}
\end{figure}

Still, even an activated link $r^j$ might fail in its transmission, leaving source and destination queue unchanged. This is modeled by a stochastic variable $m^j_t \in \{0,1\}$ which is Bernoulli distributed (coin-flip) with probability $\e{m}^j_t \in [0,1]$. I.e. $m^j_t \sim \mathcal{B}(\e{m}^j_t)$. For a succinct notation, we collect all those quantities in the diagonal matrices $M_t = \operatorname{diag}_{j = 1 ,\dots n_v}\{m^j_t\}$ and $\e{M}_t = \operatorname{diag}_{j = 1 ,\dots n_v}\{\e{m}^j_t\}$ respectively such that $M_t \sim \mathcal{B}(\e{M}_t)$ and of course $\E{M_t} = \e{M}_t$. 

The Bernoulli trials on $\e{M}_t$ are intended to model short-term stochastics. For long-term stochastics, we let $\e{M}_t$ be picked from a predetermined set $\mathcal{W} = \{ W^1,W^2,\dots \}$ of weight matrices $W^i$ according to a DTMC $(s_t)$. If $\mathcal{S} = \{1 ,\dots n_s\}$ is the index set of $\mathcal{W}$, we have $(s_t) \sim \operatorname{DTMC}(\mathcal{S},P,s_0)$ where $P$ and $s_0$ are transition matrix and initial state, respectively. The entire selection process can therefore be expressed as $\e{M}_t = W^{s_t}$. If $\sigma_t$ describes a distribution for the DTMC, we have $\lim_{t\to\infty} \sigma_t = \pi$, which we assume to be the only stable distribution, with $\pi^{\R{s}}$ being the average probability of $s_t = {\R{s}}$.

The task for a controller is to steer the packets through the network to their destination nodes. Once reached, the packets leave the system, which can be modeled via links $r^j$ without the $\{+1\}$ entry. At the same time, new packets are created directly at the queues through an \textit{arrival vector} $a_t \in \mathbb{N}^{n_q}$ of possibly stochastic nature. We call $\e{a} = \E{a_t}$ the \textit{arrival rate} and make the usual assumption, that there is an upper bound, such that always $a_t \leq \hat{a}$.

Finally, we remark that if different packets are destined for different final destination nodes, they are of different \textit{class} (or belong to a different \textit{flow}). Each class has to have its separate network of queues in order for the packets to be distinguishable. Thus, for each class a new copy of the system would have to be employed. While many authors model this by adding an additional dimension (the dimension of all classes) to all quantities, we will just assume, that the so far described system model already consists of those copies, stacked in a suitable way, thereby avoiding the introduction of another dimension to the system model.

\subsection{Control Policies and Throughput Optimality}

We already defined $q_t \in \mathcal{Q}$ and $s_t \in \mathcal{S}$. Note that $(s_t)$ is a DTMC and $q_{t-1}$ is not needed for a prediction of future states once $q_t$ is known. If we further assume the arrival vector to be independent of past realizations, there is no reason for a controller to use any but the last known realizations of $q_t$ and $s_t$ for its decision making.
If we also define the set of all control vectors by
\begin{equation}
	\mathcal{V} = \Set{ v \in \{0,1\}^{n_v} | 
		\begin{aligned}
	    C v & \leq c \\
    	\end{aligned}
    }
\end{equation}
then we can express a control policy $\phi$ as a mapping from the set of relevant, observed quantities onto the set of control vectors: $\phi: \mathcal{Q} \times \mathcal{S} \to \mathcal{V}$. In some cases, however, it makes sense to incorporate a stochastic process into the policy itself, in order to circumvent the discreteness of $\mathcal{V}$. This way, given a fixed pair of observations $(q',s')$, we can not only access a fixed $v' = \phi(q',s') \in \mathcal{V}$, but on average \textit{any} predetermined element $\sum_{v \in \mathcal{V}} \lambda^v v$ of the convex hull $\conv{\mathcal{V}}$. Hence we define a \textbf{control policy} as
\begin{equation}
	\phi: \mathcal{Q} \times \mathcal{S} \times \Omega \to \mathcal{V}
\end{equation}
where $\Omega$ is the sample space of the underlying stochastic process. (Remember that a control policy is only valid, if it complies with $\Rminus \phi(q_t,s_t,\omega_t) \leq q_t$.)

We say that a control policy $\phi$ \textbf{stabilizes} a system for a given arrival rate $\e{a}$ if it can compensate the arrival rate on average:
\begin{equation}
\label{eq::def_stability}
\begin{aligned}
	\mathbf{0} &= \lim_{\tau \to \infty} \frac{1}{\tau} \sum_{t=1}^\tau \left( \vphantom{\frac{1}{2}} a_t + R M_t \phi(q_t,s_t,\omega_t) \right)
	\\
	&= \hspace{10mm} \e{a} \hspace{10.7mm} + \sum_{\R{s} \in \mathcal{S}} \pi^\R{s} RW^\R{s} \phi(q_t,\R{s},\omega_t)
\end{aligned}
\end{equation}
Comparing with the system equation \eqref{eq::basic_queueing}, this implies that the average queue state remains bounded.

Finally, a control policy $\phi$ is throughput optimal, if it stabilizes a given system for any arrival rate $\e{a}$ for which at least one other (possibly unknown) policy stabilizes the system. This can readily be expressed by noting that a policy $v_t = \phi(s_t,\omega_t)$ can on average, for every state of $\mathcal{S}$ separately, excess any predetermined element in the interior of the convex hull $\conv{\mathcal{V}}$. Note that this excludes the boundary of $\conv{\mathcal{V}}$, because due to the positiveness constraints, no policy can guarantee to never be forced to be idle (meaning $v_t = \mathbf{0}$). Naturally, there are no other options for the average control vector than those in $\conv{\mathcal{V}}$.
Thus, $\phi$ is \textbf{throughput optimal}, if it stabilizes the system for all $\e{a}$ with
\begin{equation}
	\label{eq::def_to}
	\e{a} + \sum_{\R{s} \in \mathcal{S}} \pi^\R{s} RW^\R{s} \sum_{v \in \mathcal{V}} \lambda^{\R{s},v}  v = -\varepsilon \mathbf{1}_{n_q}
	\hspace{10mm}
	\begin{aligned}
	\lambda^{\R{s},v} & \geq 0 \\
	\sum_v \lambda^{\R{s},v} & \leq 1
	\end{aligned}
\end{equation}
where $\varepsilon > 0$ is used to exclude the boundary of $\conv{\mathcal{V}}$.

\section{Predictive Network Control ($\PNC$)}

Inspired by the common MPC paradigms, our novel control policy, $\PNC$, works in 3 steps: i) An entire trajectory of optimal control decisions from $t$ up until $t+H-1$ is calculated as the result of a minimization of an objective function $J$. Here, $J$ is a function of the next $H$ future system states, which we can only \textit{predict}. $H$ is called the prediction horizon. ii) Only the first (i.e. the immediate) control decision in this trajectory is actually applied to the system. iii) The process repeats (discarding the rest of the just calculated trajectory). W.l.o.g., for the rest of the paper, we assume the current time slot to be $t=0$.

The most often encountered objective is the sum of squares, which in our case translates to
\begin{equation}
	\label{eq::pnc_objective}
	J(q_0,s_0) = \CE{ \sum_{t=1}^{H} q_t^Tq_t }{q_0,s_0}
\end{equation}
Using this definition, minimizing $J$ means minimizing the amount of packets in the network, which can only be done by delivering the packets to their destinations.
Over the system evolution \eqref{eq::basic_queueing}, $J$ will be influenced by the choice of control vectors via
\begin{equation}
	\label{eq::actual_expected_evolution}
    \CE{q_{T}}{q_0,\sigma_0}
    = q_0 + \sum_{t = 0}^{T-1} \sum_{\R{s} \in \mathcal{S}} \left( \sigma_0 P^{(t)} \right)^\R{s} R W^\R{s}
    v_t
    + T\e{a}
\end{equation}
where $\sigma_0$ is the distribution corresponding to the initial state $s_0$ and $\left( \sigma_0 P^{(t)} \right)^\R{s}$ stands for the $\R{s}$-th entry of the predicted distribution in time slot $t$. 

Note, that the prediction can be implemented in three different ways, varying in precision and required effort:

i) The first one is the \textit{true prediction}, which assigns a control vector to \textit{every} time slot (up until $H$) and \textit{every} possible set of realizations of the quantities in the system evolution. Since the ensemble of these realizations in time slot $t$ forms $q_{t+1}$ and $s_{t+1}$, this would mean making $v_t$ a function of $q_t$ and $s_t$ for the remainder of the prediction. The number of control vectors, required for such a prediction amounts to $H \cdot n_s \cdot k_q$, where $k_q$ is the number of all possible queue states, that can be realized in a single time slot (likely to depend itself on prior queue state, realization of arrival and Bernoulli trial). This obviously requires the maximum amount of computational resources but allows us to truly find the optimal control trajectory that minimizes $J$.

ii) In contrast, a \textit{relaxed prediction} uses only a minimum of control vectors. I.e. in every time slot, a single control vector is chosen and thus $v_t$ is only a function of $t$ for the remainder of the prediction. (Using even less control vectors would not constitute a meaningful prediction for our purposes.) This amounts to only $H$ control vectors being required for the prediction, speeding up the calculation of an optimal control trajectory to minimize $J$ considerably. However, said trajectory might be sub-optimal compared to the true prediction from before and as a consequence the control performance might be worse.

iii) Finally, a mixture of both cases could be implemented, finding a balance between computational complexity and control performance. E.g. one could consider every future DTMC state $s_t$ (up until $s_H$) leading to $H \cdot n_s$ control vectors that have to be determined in order to minimize $J$.

We will define or policy via case ii), i.e. the relaxed prediction and explain the reasoning for this in the end of the section. For what follows, we substitute the control vector $v_t$ with $u_t$ to emphasize, that this is not the actual control of the queueing network but rather the one used for the prediction. Hence, $u_t$ is a quantity that is used internally to define the $\PNC$ policy according to
\begin{multline}
	\label{eq::pnc_expected_evolution}
    \CE{q_{T}}{q_0,\sigma_0}
    = q_0 + \sum_{t = 0}^{T-1} \sum_{\R{s} \in \mathcal{S}} \left( \sigma_0 P^{(t)} \right)^\R{s} R W^\R{s} u_t + T\e{a}
\end{multline}
If we define the trajectory
\begin{equation}
    \label{eq::def_w_tilde}
    \tr{u}_0^{H-1} = \begin{pmatrix}
        u_0 \\
        \vdots \\
        u_{H-1}
    \end{pmatrix}
\end{equation}
and substitute it together with \eqref{eq::pnc_expected_evolution} into \eqref{eq::pnc_objective}, the objective can be rewritten as
\begin{equation*}
    J(q_0,s_0) = J_1(q_0) + J_2(q_0,s_0) \tr{u}_0^{H-1} + \left( {\tr{u}_0^{H-1}} \right)^\intercal J_3(s_0) \tr{u}_0^{H-1}
\end{equation*}
As can be seen, $J_1(q_0)$ will not be influenced by the minimization over $\tr{u}_0^{H-1}$, and $J_3(s_0)$ will stay bounded, since it is not dependent on $q_0$. Therefore, for large enough $q_0$, the linear term $J_2(q_0,s_0) \tr{u}_0^{H-1}$ will always dominate the minimization over $\tr{u}_0^{H-1}$, making it prudent to define our actual objective only over this linear term. This simplifies the minimization from a quadratic to a linear one. (Note that a similar step is also taken in the definition of the $\MW$ policy.) Expanding the remaining constraints from the original system in a straight forward way, we end up with the following definition of the $\PNC$ policy:
\begin{equation}
	\label{eq::pnc_policy_def}
\begin{gathered}
    \phi^{\PNC}(q_0,s_0) = \operatorname{first} \argmin_{\tr{u}_0^{H-1}} \  J_2(q_0,s_0) \tr{u}_0^{H-1}
    \\[2ex]
    \text{subject to} \qquad
    \begin{aligned}
    \tr{C} \tr{u}_0^{H-1} &\leq \tr{c}
    \\
    \tr{A} \tr{u}_0^{H-1} &\leq \tr{b}(q_0)
    \\
    \tr{u}_0^{H-1} &\in \{0,1\}^{Hn_v}
    \end{aligned}
\end{gathered}
\end{equation}
where $\operatorname{first} \argmin ()$ expresses that only the first argument of the trajectory is used as output. An overview of the utilized quantities can be found in Table~\ref{table_01}.

Choosing $H=1$, we end up with the common $\MW$ policy, which is not surprising, since its definition also involves a quadratic objective. And indeed, $\PNC$ would merely be the extension of $\MW$ over multiple time slots, if the $\PNC$ controller would follow a once calculated optimal trajectory to its end (i.e. for $H$ time slots). However, $\PNC$ recalculates this trajectory each time slot, thereby discarding its entire tail. This results in a much improved behavior of $\PNC$ (see section \ref{subsec::synchronized_queues}) but also makes it impossible to infer any properties from $\MW$ to $\PNC$. For more comparisons between the two policies, we refer to \cite{Schoeffauer2018a} and \cite{Schoeffauer2019}.

We continue with the main theorem of this paper, which states throughput optimality of the $\PNC$ policy. Note that this automatically implies throughput optimality for every other MPC controller, that uses a more precise prediction (under the same constraints and objective function).
\begin{theorem}
\label{theorem::to_of_pnc}
The $\PNC$ policy \eqref{eq::pnc_policy_def} is throughput optimal for the system \eqref{eq::basic_queueing}.
\end{theorem}

\begin{table*}[ht!]
\caption{Extended Formulas for the Optimization Problem}
\centering
\normalsize
\begin{tabular}{|p{0.3\linewidth}|p{0.64\linewidth}|}
\hline
\hline
\begin{center}
Expected value of weight matrix $W^{s_t}$
\end{center}
&
\begin{center}
Linear objective $J_2$
\end{center}
\\
\vspace{-2.8mm}
\begin{equation}
\begin{gathered}
    \e{W}_t(s_0) = 
	\CE{W^{s_t}}{s_0} 
	\\  =
    \left( \sigma_0 P^{(t)} \otimes I_{n_s} \right)
    \begin{pmatrix}
		W^1 \\ W^2 \\ \vdots \\ W^{n_s}
    \end{pmatrix}^\intercal
    \end{gathered}
\end{equation}
&
\begin{equation}
    \label{eq::J2_formula}
    J_2 = 
    2q_0^\intercal R
	\begin{pmatrix}
		H \e{W}_0(s_0)
		\\
		(H-1) \e{W}_1(s_0)
		\\
		\vdots
		\\
		 \e{W}_{H-1}(s_0)
	\end{pmatrix}^\intercal
	+ \e{a}^\intercal R
	\begin{pmatrix}
		(H+1)(H-0) \e{W}_0(s_0)
		\\
		(H+2)(H-1) \e{W}_1(s_0)
		\\
		\vdots
		\\
		2H \e{W}_{H-1}(s_0)
	\end{pmatrix}^\intercal
\end{equation}
\\
\hline
\begin{center}
Constituency constraints
\end{center}
&
\begin{center}
Positiveness constraints
\end{center}
\\
\begin{equation}
	\label{eq::table_constituency}
\underbrace{
    		\left( \vphantom{\frac{1}{2}} I_{H} \otimes C \right) 
	}_{\displaystyle \tr{C} }  	
 	\tr{u}_0^{H-1} \leq
	\underbrace{ 	
 		\mathbf{1}_{H} \otimes c \vphantom{\left( \vphantom{\frac{1}{2}} I_{H} \otimes C \right) }
 	}_{\displaystyle \tr{c} \vphantom{\tr{C}} }
\end{equation}
&
\vspace{-2.6mm}
\begin{equation}
	\label{eq::table_positiveness}
    \underbrace{
    \begin{pmatrix}
        \Rminus &  & &  \\
        R & \Rminus &   &  \\
        \vdots &  & \ddots \\
        R & \dots & R & \Rminus
    \end{pmatrix}
    }_{\displaystyle \tr{A}}
    \tr{u}_0^{H-1}
    \leq 
    \underbrace{
    \begin{pmatrix}
    	q_0 \\ q_0 + \e{a} \\ \vdots \\ q_0 + (H-1) \e{a}
    \end{pmatrix}
    }_{\displaystyle \tr{b}(q_0) }
\end{equation}
\\
\hline
\hline
\end{tabular}
\label{table_01}
\end{table*}

\section{Proof of Theorem \ref{theorem::to_of_pnc}}
\label{sec::proof}

We will now prove, that $\phi^\PNC$ is throughput optimal. And in contrast to the usual stability-related proofs employed for MPC controllers, we will not rely on a terminal set.

\subsection{Preliminaries}

It will often become necessary to upper and lower-bound certain expressions. We will use $K_i \in \mathbb{R}_+$, $i\in \mathbb{N}$ to denote those bounds or variables, whose values are of no further interest and are obvious to calculate. Crucially, any $K_i$ will be independent of the initial system state $q_0$!

We will use gothic letters to express realizations of random variables, such that e.g. $\R{s_t}$ is a realization of $s_t$, hence $\R{s_t} \in \mathcal{S}$. And because the corresponding set of realizations will always be very clear from the context, we will use the succinct notation $\sum_{\R{s_t}}$ instead of $\sum_{\R{s_t} \in \mathcal{S}}$ for the sum of all realizations (as is needed for e.g. expressing the expectation).

Given a trajectory $\tr{x}$ of control vectors of certain length, we use $\tr{x} \in \mathcal{P}(q_0)$, to express that $\tr{x}$ abides to the positiveness constraints $\tr{A} \tr{x} \leq \tr{b}(q_0)$, where $\tr{A}$ and $\tr{b}(q_0)$ are defined as in \eqref{eq::table_positiveness}, expect for a possibly different value of $H$ (depending on the length of $\tr{x}$). Analogue, $\tr{x} \in \mathcal{C}$ will express, that $\tr{x}$ abides to the constituency constraints as in \eqref{eq::table_constituency}.

Finally, we make the definitions $\Delta_0^T := q_T - q_0$. Keep in mind, that analogue to $q_t$ being a function of all prior stochastics and controls, $\Delta_0^T$ is of course a function of all stochastics and controls in the time slot $0,\dots T-1$. With the definition $\norm{q_t} := q_t^\intercal q_t$ (which is not meant to be a norm) this gives raise to
\begin{equation}
	\label{eq::to_sub_to_q0}
	\norm{q_T} = \norm{\D_0^T} + \norm{q_0} + 2q_0^\intercal \D_0^T
\end{equation}
Note, that this decomposition corresponds to the one for the objective function $J$ and we can now restate the objective of the $\PNC$ policy, $J_2$, as
\begin{equation}
	\label{eq::to_objective_better}
	J_2(q_0,s_0) = \CE{ \sum_{t=1}^{H} 2q_0^\intercal \D_0^t }{q_0,s_0}
\end{equation}
With this notation we formulate the next lemmas, needed for the proof.

\begin{lemma}
The difference $\Delta_0^T$ between two queue states is bounded (element-wise) by
\begin{equation}
	\label{eq::to_gen_bounds}
	- T n_v \mathbf{1}_{n_q}
	\leq
	\D_0^T
	\leq
	T n_v \mathbf{1}_{n_q}  + T \hat{a}
\end{equation}
leading to
\begin{equation}
	\label{eq::to_quad_to_lin}
	\norm{q_0} + 2q_0^\intercal  \D_0^T
	\leq
	\norm{q_T} 
	\leq \norm{q_0} + 2q_0^\intercal  \D_0^T + K_1
\end{equation}
\end{lemma}
\begin{proof}
Between time slots $0$ and $T$ we have at best a constant efflux of $n_v$ or at worst a constant influx of $n_v + \hat{a}$ packets per queue per time slot (since there are at most $n_v$ links to fill or drain any given queue).
\end{proof}

\begin{lemma}
\label{lemma::diff}
The difference between the minimization that originates from the definition of the $\PNC$ policy (using the formulation from \eqref{eq::to_objective_better}), and the same minimization but without considering any positiveness constraints can be bounded by
\begin{gather}
\label{eq::to_lemma}
\begin{aligned}
    \min_{ \tr{u}_0^{H-1} \in \mathcal{C} \cap \mathcal{P}(q_0) }
    	&\CE{ 
    		\sum_{t=1}^{H}  2q_0^\intercal \D_0^t
    	}
    	{
    		q_0,s_0
    	}
    \\
    -
    \min_{ \tr{u}_0^{H} \in \mathcal{C} } \hspace{6mm}
    	&\CE{ 
    		\sum_{t=1}^{H}  2q_0^\intercal \D_0^t
    	}
    	{
    		q_0,s_0
    	}    
\end{aligned}
    \\ \nonumber
    \leq
    \left( Hn_v^2-n_v \right) \left( H + 1 \right)
    = K_2
\end{gather}
\end{lemma}
\begin{proof}
Clearly, the maximum deficit, that $\D_0^t$ can generate is (element-wise) less than $t n_v \mathbf{1}_{n_q}$ (all links drain a queue over $t$ steps).
For a single queue, the \textit{most} efflux in $H$ time slots is $Hn_v$ packets. 
Hence, if $q_0 \geq Hn_v \mathbf{1}_{n_q}$, a control trajectory \textit{cannot} violate the positiveness. Conversely, if a link cannot be activated due to the positiveness constraints, at least one entry of $q_0$ must be smaller than $H n_v$.

Due to the linearity, the largest difference in the minimizations will be found, if $q_0 = (H n_v - 1) \mathbf{1}_{n_q}$ (possibly denying any activation for the minimization with the positiveness constraints). This, together with the initial bound on $\D_0^t$ leads directly to
\begin{equation}
	\sum_{t=1}^H 2 (H n_v - 1) \mathbf{1}_{n_q}^\intercal  t n_v \mathbf{1}_{n_q} = \left( Hn_v^2-n_v \right) \left( H + 1 \right)
\end{equation}
which is an upper bound for the difference in question.
\end{proof}

Finally, the following theorem will allow us to express our definition of stability by the means of a Ljapunov function.
\begin{lemma}
	\label{lemma::foster}
	A policy $\phi$ stabilizes the system \eqref{eq::basic_queueing} under a certain arrival rate $\e{a}$, if we can find a function $f: \mathcal{Q} \times \mathcal{S} \to \mathbb{R}_+$ with the property
	\begin{equation}
		\label{eq::to_drift}
    	\CE{ f(q_{1},s_{1}) - f(q_0,s_0) }{q_0} \leq K_4 - K_5 \mathbf{1}_{n_q}^\intercal q_0
	\end{equation}
\end{lemma}
\begin{proof}
Without further ado, we take expectations and sum \eqref{eq::to_drift} over multiple time slots to obtain the following sequence of arguments
\begin{gather}
	\nonumber
    \E{f(q_T,s_T)} - \E{f(q_0,s_0)}
    \leq T K_4 - K_5 \sum_{t = 0}^{T - 1} \E{ \mathbf{1}_{n_q}^\intercal q_t }
\\
	\nonumber
    \Longrightarrow \qquad
    -\E{f(q_0,s_0)} 
    \leq T K_4 - K_5 \sum_{t = 0}^{T - 1} \E{ \mathbf{1}_{n_q}^\intercal q_t }
	\hphantom{\Longrightarrow \qquad}    
\\
    \Longrightarrow \qquad
    \frac{1}{T} \sum_{t=0}^{T-1} \E{ \mathbf{1}_{n_q}^\intercal q_t } 
    \leq
    \frac{K_4}{K_5} + \frac{\E{f(q_0)}}{T K_5}
    \hphantom{\Longrightarrow }
\\
	\nonumber
    \Longrightarrow \qquad
    \lim_{T\to\infty}
    \frac{1}{T} \sum_{t=0}^{T-1} \E{ q_t } \leq
    \frac{K_4}{K_5} \mathbf{1}_{n_q} 
    \hphantom{\Longrightarrow \qquad}
\end{gather}
Since $\E{q_t} \geq \mathbf{0}$ always, and the difference between consecutive states is bounded, all elements of the sequence $(\E{q_t})$ must be bounded. From there, the stability condition \eqref{eq::def_stability} follows immediately. (Note that $\E{\cdot}$ is the expectation of the stochastic in the system model and not over time.)
\end{proof}

\subsection{Main Proof}

We now start with the main part of the proof. We will define a Ljapunov function $f(q_t,s_t)$ and show that \textit{if} the system is governed by the $\PNC$ policy, $f$ fulfills Lemma \ref{lemma::foster} for any possibly stabilizable arrival rate (see \eqref{eq::def_to}).

For a $\PNC$ policy with horizon $H+1$, we employ the following Ljapunov function:
\begin{equation}
    \label{eq::to_f_tilde}
\begin{gathered}    
    f(q_0,s_0) =
    \min_{ \tr{z}_0^{H-1} \in \mathcal{C}_z } 
    \CE{ 
    		\sum_{t=1}^{H} \norm{q_{t}}
    }{q_0,s_0}
\end{gathered}
\end{equation}
A few remarks are in order: i) The minimization in $f$ mimics the $\PNC$ policy, but is in fact independent of it. ii) The horizon of the minimization of $f$ is chosen to be one step smaller than that of the $\PNC$ policy. iii) The control vectors are now denoted by $z$ instead of $v$ or $u$, because they run independent of the actual control $v_t$ or the predicted control inside the $\PNC$ controller $u_t$. Crucially, the control trajectory $\tr{z}_0^{H-1}$ is state sensitive regarding the DTMC $(s_t)$ of the weight matrices and is \textit{not} constraint by the positiveness constraints $\mathcal{P}$. I.e.
\begin{multline}
	\label{eq::to_expected_evolution}
    \CE{q_{T}}{q_0,\sigma_0}
    = q_0 + \sum_{t = 0}^{T-1} \sum_{\R{s} \in \mathcal{S}} \left( \sigma_0 P^{(t)} \right)^\R{s} R W^\R{s} z_t(\R{s}) + T\e{a}
\end{multline}

This last point is important: the minimization of the actual $\PNC$ policy (with horizon $H+1$) uses the control trajectory $\tr{u}_0^{H}$, which assigns to each time slot of the prediction \textit{a single} control vector $u_t \in \mathcal{V}$. In contrast and per definition, the minimization in the Ljapunov function $f$ uses the control trajectory $\tr{z}_0^{H-1}$, which assigns to each time slot \textit{and} each possible realization of $s_t$ a control vector $z_t(s_t) \in \mathcal{V}$. For succinct notation we define $z_t^\R{s} := z_t(s_t = \R{s})$. The trajectory of the control vectors $z_t^\R{s}$ is defined by first stacking over all realization of $(s_t)$ and then over all time slots:
\begin{equation}
	\tr{z}_0^{H} 
	=
	\begin{pmatrix}
		z'_0 \\
		z'_1 \\
		\vdots \\
		z'_H
	\end{pmatrix}
	,\qquad \text{with} \qquad
	z'_t =
	\begin{pmatrix}
		z_t^1 \\
		z_t^2 \\
		\vdots \\
		z_t^{n_s}
	\end{pmatrix}
\end{equation}

The constraint $\tr{z}_0^{H-1} \in \mathcal{C}_z$ in the definition of $f$ expresses, that each single control vector $z_t^\R{s}$ is constraint by the constituency in the usual way ($C z_t^\R{s} \leq c$), and therefore $\mathcal{C}_z$ is a straight forward expansion of $\mathcal{C}$.

We can now start expressing the first term of \eqref{eq::to_drift} (for now conditioning on $s_0$ as well) as
\begin{gather}
	\label{eq::to_simple_first_term}
\CE{ f (q_1,s_1) }{ q_0,s_0 }	 =
	\CE{
		\min_{ \tr{z}_1^{H} \in \mathcal{C} }
		\CE{
			\sum_{t=2}^{H+1} \norm{q_t}
		}{
			q_1 , s_1
		}
	}{
		q_0,s_0
	}
\end{gather}
Note that this term is shifted in time. The control $v_0$, that leads from $q_0$ to $q_1$ is exactly the control, that is defined by the control policy $v_0 = \phi^\PNC$ and that actually affects the network. In contrast, the dummy controls $\tr{z}_1^H$ are part of the function $f$, do not affect the actual network and therefore are independent of the chosen policy.

Using \eqref{eq::to_sub_to_q0} the inner term of \eqref{eq::to_simple_first_term} can become
\begin{align*}
	\sum_{t=2}^{H+1} \norm{q_t}
	&=
	\sum_{t=2}^{H+1} \Big(
		\norm{q_0} + \norm{\D_0^t} + 2q_0^\intercal \D_0^t
	\Big)
	\\
	&\leq
	H \norm{q_0} + K_6
	+
	\sum_{t=2}^{H+1} \Big(
		2q_0^\intercal \D_0^t
	\Big)
\end{align*}
The individual sums of $\D_0^t$ can be bounded according to \eqref{eq::to_gen_bounds} by some constant $K_6$ which is unaffected by the minimization or the expectation from \eqref{eq::to_simple_first_term}. The same holds for $\norm{q_0}$ if we notice, that conditioning on $q_1$ is the same as conditioning on $q_0$ \textit{and} $\D_0^1$, since $q_1 = q_0 + \D_0^1$. Hence, both terms can be pulled to the left-hand-side of \eqref{eq::to_simple_first_term}, as seen in the first two lines of \eqref{eq::to_long}. In what follows, we will step by step dissolve the outer expectation and expand the sum, which is possible, since the minimization is linear in $\tr{z}_1^H$ and the constraints act on each control vector separately:
\begin{align}
	\label{eq::to_long}
	& \hspace{5mm}
	\CE{ f (q_1,s_1) }{ q_0,s_0 } - H \norm{q_0} - K_6 
\\ 
	\nonumber
	& \leq 
	\mathbb{E} \Bigg[
		\min_{ \tr{z}_1^{H} \in \mathcal{C}_z }
		\mathbb{E} \Bigg[
			\sum_{t=2}^{H+1}			
				2q_0^\intercal \D_1^t
		\, \Bigg| \,
			q_1 , s_1
		\Bigg]
	\, \Bigg| \,
		q_0,s_0
	\Bigg]
\\ 
	\nonumber
	& = 
	\mathbb{E} \Bigg[
		\min_{ \tr{z}_1^{H} \in \mathcal{C}_z}
		\mathbb{E} \Bigg[
			\sum_{t=2}^{H+1}			
				2q_0^\intercal \D_1^t
		\, \Bigg| \,
			q_0 ,\D_0^1, s_1
		\Bigg]
	\, \Bigg| \,
		q_0,s_0
	\Bigg]
\\ 
	\nonumber
	& = 
	\sum_{\R{q_0},\R{s_1}}
		\min_{ \tr{z}_1^{H} \in \mathcal{C}_z }
		\mathbb{E} \Bigg[
			\sum_{t=2}^{H+1}			
				2q_0^\intercal  \D_1^t
		\, \Bigg| \,
			\R{q_0} , \R{s_1}
		\Bigg]
	\CP{
		\R{q_0} , \R{s_1}
	}{
		q_0,s_0
	}
\\
	\nonumber
	& \hspace{1mm}
	\begin{aligned}
	=
	\sum_{\R{q_0},\R{s_1}}
		\min_{ \tr{z}_1^{H} \in \mathcal{C}_z }
			\mathbb{E} \Bigg[
			\sum_{t=1}^{H}
			\sum_{\tau = t}^{H}			
				2q_0^\intercal  \left(
				R M_t z_t^{s_t} + a_t \right)
		\, \Bigg| \,
			\R{q_0} , \R{s_1}
		\Bigg] &
	\\
	\cdot	
	\CP{
		\R{q_0}
	}{
		q_0
	}
	\CP{
		\R{s_1}
	}{
		s_0
	} \, &
	\end{aligned}
\\ 
	\nonumber
	& \hspace{1mm}
	\begin{aligned}
	=
	\sum_{\R{s_1}}
		\min_{ \tr{z}_1^{H} \in \mathcal{C}_z }
			\mathbb{E} \Bigg[
			\sum_{t=1}^{H}
			\sum_{\tau = t}^{H}				
				2q_0^\intercal  \left(
				R M_t z_t^{s_t} + a_t \right)
		\, \Bigg| \,
			q_0 , \R{s_1}
		\Bigg] &
	\\
	\cdot
	\CP{
		\R{s_1}
	}{
		s_0
	} \, &
	\end{aligned}
\\ 
	\nonumber
	& \hspace{1mm}
	\begin{aligned}
	=
	\sum_{\R{s_1}}
		\min_{ \tr{z}_1^{H} \in \mathcal{C}_z }
			\sum_{t=1}^{H}
			\sum_{\tau = t}^{H}	
			\sum_{\R{s_t}}			
				2 q_0^\intercal  \left(
				R W^\R{s_t} z_t^\R{s_t} + a_t \right)
		\CP{ \R{s_t} }{\R{s_1}} &
	\\[-2ex]
	\cdot
	\CP{
		\R{s_1}
	}{
		s_0
	} &
	\end{aligned}
\\ 
	\nonumber
	& \hspace{1mm}
	\begin{aligned}
	=
		\min_{ \tr{z}_1^{H} \in \mathcal{C}_z }
			\sum_{\R{s_1}}
			\sum_{t=1}^{H}
			\sum_{\tau = t}^{H}	
			\sum_{\R{s_t}}			
				2q_0^\intercal  \left(
				R W^\R{s_t} z_t^\R{s_t} + a_t \right)
		\CP{ \R{s_t} }{\R{s_1}} &
	\\[-2ex]
	\cdot
	\CP{
		\R{s_1}
	}{
		s_0
	} &
	\end{aligned}
\\ 
	\nonumber
	& = 
		\min_{ \tr{z}_1^{H} \in \mathcal{C}_z }
			\sum_{t=1}^{H}
			\sum_{\tau = t}^{H}	
			\sum_{\R{s_t}}			
				2q_0^\intercal   \left(
				R W^\R{s_t} z_t^\R{s_t} + a_t \right)
		\CP{ \R{s_t} }{s_0}
\\ 
	\nonumber
	& = 
		\min_{ \tr{z}_1^{H} \in \mathcal{C}_z }
		\mathbb{E} \Bigg[
			\sum_{t=1}^{H}
			\sum_{\tau = t}^{H}	
				2q_0^\intercal   \left(
				R M_{t} z_t^{s_t} + a_t \right)
		\, \Bigg| \,
			q_0,s_0
		\Bigg]
\\ 
	\nonumber
	& = 
		\min_{ \tr{z}_1^{H} \in \mathcal{C}_z }
		\mathbb{E} \Bigg[
			\sum_{t=2}^{H+1}
				2q_0^\intercal  
				\D_1^t
		\, \Bigg| \,
			q_0,s_0
		\Bigg]
\\ 
	\nonumber
	& =
		\min_{ \tr{z}_1^{H} \in \mathcal{C}_z }
		\mathbb{E} \Bigg[
			\sum_{t=1}^{H+1}
				2q_0^\intercal  
				\D_0^t
				- 2q_0^\intercal  \D_0^1
		\, \Bigg| \,
			q_0,s_0
		\Bigg]
\\ 
    \label{eq::to_last_long}
	& \leq
		\min_{ \tr{z}_1^{H} \in \mathcal{C}_z}
		\mathbb{E} \Bigg[
			\sum_{t=1}^{H+1}
				2q_0^\intercal  
				\D_0^t
		\, \Bigg| \,
			q_0,s_0
		\Bigg]
	-
	\min_{ \tr{z}_0^{0} \in \hphantom{\mathcal{C}} \mathclap{\mathcal{C}_z} }
	\mathbb{E} \Bigg[	
	2q_0^\intercal  \D_0^1
	\, \Bigg| \,
		q_0,s_0
	\Bigg]
\end{align}
Though $\D_0^1$ is steered by $v_0$ (the actual control of the system), every policy has to abide to the constituency, which allows for the last term to be formulated over $\tr{z}_0^0 \Leftrightarrow z_0^{s_0}$. Still, the first term of \eqref{eq::to_last_long} depends on $v_0$ through $\D_0^t$ so that we can rewrite it as
\begin{multline}
\label{eq::to_elim_a}
	\min_{ \tr{z}_1^{H} \in \mathcal{C}_z }
	\mathbb{E} \Bigg[
		\sum_{t=1}^{H+1}
			2q_0^\intercal 
			\D_0^t
	\, \Bigg| \,
		q_0,s_0
	\Bigg]
	=
	\frac{ \left( H+1 \right) \left( H+2 \right) }{2} \e{a}
\\
		+
		\min_{ \tr{z}_1^{H} \in \mathcal{C}_z }
		2q_0^\intercal R \left[
		    \sum_0^H
		        W^{s_0} v_0
		    +
		    \sum_{t=1}^H \sum_{\tau=t}^{H} \sum_{\R{s_t}} W^\R{s_t} z_t^\R{s_t} \CP{\R{s_t}}{s_0}
		\right]
\end{multline}

To interface the $\PNC$ policy \eqref{eq::pnc_policy_def}, we need to incorporate the positiveness constraints. To that end, we introduce a transformation in variables, centered around the idea, that each set $\{ z_t^1,\dots z_t^{n_s} \}$ can be expressed by a common part $\mu_t$, and $n_s$ differences $\delta_t^i$:
\begin{equation}
\begin{gathered}
	z_t^i = \mu_t + \delta_t^i \in \mathcal{C}_z
	\\
	\text{for} \qquad  t = 1,\dots H \qquad \text{and} \qquad i = 1 ,\dots n_s
	\\
	\text{with}  \qquad  \mu_t,\delta_t^i \in \{0,1\}^{n_v} \qquad \text{and} \qquad C \mu_t \leq c 
\end{gathered}
\end{equation}
We define a suitable stacking of these new variables in such a way that
$\tr{z}_1^H = \tr{\mu}_1^H + \tr{\delta}_1^H$ and write $\tr{\delta}_1^H \in \mathcal{C}_z \setminus \tr{\mu}_1^H $, to express that $\tr{\delta}_1^H$ has to abide to usual constituency, if $\tr{\mu}_1^H$ has already been chosen (i.e. for each $\delta_t^i$ separately it must hold that $C \delta_t^i \leq c - C \mu_t$).

Next, we substitute these variables into the last term of \eqref{eq::to_elim_a}.
\def\hsp{\hspace{1cm}}
\def\minsp{\hspace{3mm}}
\begingroup
\allowdisplaybreaks
\begin{align}
    \nonumber
    &
	\min_{ \tr{z}_1^{H} \in \mathcal{C}_z }
		2q_0^\intercal  R 
		\Bigg[
		\sum_{0}^{H}  W^{s_0} v_0
		 +
		\sum_{t=1}^H \sum_{\tau=t}^{H} \sum_{\R{s_t}} W^\R{s_t} z_t^\R{s_t} \CP{\R{s_t}}{s_0}
		\Bigg]
\\ \nonumber
	&
	=
	\min_{ \tr{\mu}_1^H + \tr{\delta}_1^H \in \mathcal{C}_z }
		2q_0^\intercal  R 
		\Bigg[
		\sum_{0}^{H} W^{s_0} v_0
		\\ \nonumber
		& \hsp +
		\sum_{t=1}^H \sum_{\tau=t}^{H} \sum_{\R{s_t}} W^\R{s_t} \left( \mu_t + \delta_t^\R{s_t} \right) \CP{\R{s_t}}{s_0}
		\Bigg]
\\ \nonumber
	&
	=
	\min_{ \tr{\delta}_1^H \in \mathcal{C}_z \setminus \tr{\mu}_1^H }
	\minsp
	\min_{ \vphantom{\tr{\delta}_1^H} \tr{\mu}_1^H \in \mathcal{C} }
		2q_0^\intercal R 
		\Bigg[
		\sum_{0}^{H} W^{s_0} v_0
		\\ \nonumber
		& \hsp +
		\sum_{t=1}^H \sum_{\tau=t}^{H} \left( \sum_{\R{s_t}} W^\R{s_t} \CP{\R{s_t}}{s_0} \right) \mu_t 
		\\ \nonumber
		& \hsp +
		\sum_{t=1}^H \sum_{\tau=t}^{H} \sum_{\R{s_t}} W^\R{s_t} \delta_t^\R{s_t} \CP{\R{s_t}}{s_0}
		\Bigg]
\\ \nonumber
	&
	=
	\min_{ \tr{\delta}_1^H \in \mathcal{C}_z \setminus \tr{\mu}_1^H }
	\minsp
	\min_{ \vphantom{\tr{\delta}_1^H} \tr{\mu}_1^H \in \mathcal{C} }
		2q_0^\intercal R 
		\Bigg[
		\sum_{0}^{H} W^{s_0} v_0
		\\ \nonumber
		& \hsp +
		\sum_{t=1}^H \sum_{\tau=t}^{H} \e{W}_t(s_0) \mu_t
		\\ \nonumber	
		& \hsp +		
		\sum_{t=1}^H \sum_{\tau=t}^{H} \sum_{\R{s_t}} W^\R{s_t} \delta_t^\R{s_t} \CP{\R{s_t}}{s_0}
		\Bigg]
\\ \nonumber
	&
	\leq
	\min_{ \tr{\delta}_1^H \in \mathcal{C}_z \setminus \tr{\mu}_1^H }
	\minsp
	\min_{ \vphantom{\tr{\delta}_1^H}  \left( \mathbf{0}_{n_v}^\intercal , \left( \tr{\mu}_1^H\right )^\intercal \right)^\intercal \in \mathcal{C} \cap \mathcal{P}(q_0)}
		2q_0^\intercal R 
		\Bigg[
		\cdot
		\Bigg]
\\ \nonumber
	&
	\overset{\mathclap{\PNC}}{=}
	\min_{ \tr{\delta}_1^H \in \mathcal{C}_z \setminus \tr{\mu}_1^H }
	\minsp
	\min_{ \vphantom{\tr{\delta}_1^H} \left( v_0^\intercal , \left( \tr{\mu}_1^H\right )^\intercal \right)^\intercal \in \mathcal{C} \cap \mathcal{P}(q_0) }
		2q_0^\intercal R 
		\Bigg[
		\cdot
		\Bigg]
\\ \nonumber
	&
	\overset{\mathclap{\text{Lemma \ref{lemma::diff}}}}{
	\leq
	}
	K_7 +
	\min_{ \tr{\delta}_1^H \in \mathcal{C}_z \setminus \tr{\mu}_1^H }
	\minsp
	\min_{ \vphantom{\tr{\delta}_1^H} \left( v_0^\intercal , \left( \tr{\mu}_1^H\right )^\intercal \right)^\intercal \in \mathcal{C} }
		2q_0^\intercal R 
		\Bigg[
		\cdot
		\Bigg]
\\ \nonumber
	&
	\leq
	K_7 +
	\min_{ \tr{\delta}_1^H \in \mathcal{C}_z \setminus \tr{\mu}_1^H }
	\minsp
	\min_{ \vphantom{\tr{\delta}_1^H}  v_0 \in \mathcal{C} }
	\minsp
	\min_{ \vphantom{\tr{\delta}_1^H} \tr{\mu}_1^H \in \mathcal{C} }
		2q_0^\intercal R 
		\Bigg[
		\cdot
		\Bigg]
\\ \nonumber
	&
	=
	K_7 +
	\min_{ \tr{z}_1^H \in \mathcal{C}_z }
	\minsp
	\min_{ v_0 \in \mathcal{C} }
		2q_0^\intercal R 
		\Bigg[
		\sum_{0}^{H} W^{s_0} v_0
\\ \nonumber
		& \quad +		
		\sum_{t=1}^H \sum_{\tau=t}^{H} \sum_{\R{s_t}} W^\R{s_t} z_t^\R{s_t} \CP{\R{s_t}}{s_0}
		\Bigg]
\\ \nonumber
	&
	=
	K_7 + 
	\min_{ \tr{z}_0^H \in \mathcal{C}_z }
	2q_0^\intercal R 
		\sum_{t=0}^H \sum_{\tau=t}^{H} \sum_{\R{s_t}} W^\R{s_t} z_t^\R{s_t} \CP{\R{s_t}}{s_0}
\\
    \label{eq::to_intro_PNC}
	&
	=
	K_7 + 
	\min_{ \tr{z}_0^H \in \mathcal{C}_z }
		\CE{ \sum_{t=1}^{H+1} 2q_0^\intercal R \D_0^t }{q_0,s_0}
\end{align}
Note that we used the fact that $\mu_t$, once separated from the $\delta_t^i$ in a suitable manner, can be identified with the dummy control $u_t$ from the definition of the $\PNC$ policy \eqref{eq::pnc_objective}. Hence, the equality marked with the $\PNC$ label \textit{only} holds under the $\PNC$ policy, since it chooses $v_0$ in such a way that the entire first term is minimized over the trajectory $\tr{z}_0^H$ instead of only $\tr{z}_1^H$.

If we now combine the results of \eqref{eq::to_long}, \eqref{eq::to_elim_a} and \eqref{eq::to_intro_PNC} we get
\begin{gather}
	\label{eq::to_first_drift_term}
	\CE{ f (q_1,s_1) }{ q_0,s_0 } - H \norm{q_0} - K_8 
	\\
	\nonumber
	\leq
	\min_{ \tr{z}_0^{H} \in \hphantom{\mathcal{C}} \mathclap{\mathcal{C}_z} }
	\mathbb{E} \Bigg[
		\sum_{t=1}^{H+1}   2q_0^\intercal  \D_0^t
	\, \Bigg| \,
		q_0,s_0
	\Bigg]
	-
	\min_{ \tr{z}_0^{0} \in \hphantom{\mathcal{C}} \mathclap{\mathcal{C}_z} }
	\mathbb{E} \Bigg[	
	2q_0^\intercal  \D_0^1
	\, \Bigg| \,
		q_0,s_0
	\Bigg]
\end{gather}
which completes the derivation for the first term of \eqref{eq::to_drift}.

In a similar but much easier fashion, the second term of \eqref{eq::to_drift} can be reshaped into
\begin{equation}
	\label{eq::to_second_drift_term}
\begin{gathered}
	\CE{ f (q_0,s_0) }{ q_0,s_0 } - H \norm{q_0}
	\\
	\geq
	\min_{ \tr{z}_0^{H-1} \in \mathcal{C}_z }
	\mathbb{E} \Bigg[
		\sum_{t=1}^{H}   2q_0^\intercal  \D_0^t
	\, \Bigg| \,
		q_0,s_0
	\Bigg]
\end{gathered}
\end{equation}
Combining \eqref{eq::to_first_drift_term} and \eqref{eq::to_second_drift_term} results in
\begin{gather*}
	\CE{ f (q_1,s_1) - f (q_0,s_0) }{ q_0,s_0 }
\\
	\leq
	K_{8} +
	\min_{ \tr{z}_1^{H} \in \mathcal{C}_z }
	\mathbb{E} \Bigg[
		2q_0^\intercal  \D_1^{H+1}
	\, \Bigg| \,
		q_0,s_0
	\Bigg]
\end{gather*}
To alleviate the outer conditioning on $s_0$ we take the expectation $\CE{\cdot}{q_0}$ on both sides, conditioned only on $q_0$, and swap minimization and expectation operator:
\begin{gather*}
	\CE{ f (q_1,s_1) - f (q_0,s_0) }{ q_0 }
\\
	\leq
	K_{9} +
	\mathbb{E} \Bigg[
	\min_{ \tr{z}_1^{H} \in \mathcal{C}_z }
	\mathbb{E} \Bigg[
		2q_0^\intercal  \D_1^{H+1}
	\, \Bigg| \,
		q_0,s_0
	\Bigg]
	\, \Bigg| \,
		q_0
	\Bigg]
\\
	\leq
	K_{9} +
	\min_{ \tr{z}_1^{H} \in \mathcal{C}_z }
	\mathbb{E} \Bigg[
		2q_0^\intercal  \D_1^{H+1}
	\, \Bigg| \,
		q_0
	\Bigg]
\end{gather*}
Finally, recall that for throughput optimality, this expression has to be negative for each $\e{a}$ that can be expressed via \eqref{eq::def_to}. Substituting this we obtain
\begin{gather*}
	\CE{ f (q_1,s_1) - f (q_0,s_0) }{ q_0 }
\\
\begin{aligned}
	& \leq
	K_{9} +
	\min_{ \tr{z}_1^{H} \in \mathcal{C}_z }
		2q_0^\intercal \left(
		\sum_{t=1}^{H}
			\sum_{\R{s}}  \pi^\R{s} R  W^\R{s} z_t^\R{s} + H \e{a}
		\right)
\\ &
\begin{multlined}
	\leq
	K_{9} +
	\min_{ z' \in \mathcal{C}_z }
		2Hq_0^\intercal \Bigg(
		\sum_{\R{s}} \pi^\R{s} RW^\R{s} z^\R{s}
	\\
		-  \varepsilon \mathbf{1}_{n_q}
		-  \sum_{\R{s}} \pi^\R{s} R W^\R{s} \sum_{v \in \mathcal{V}} \lambda^{\R{s},v} v
		\Bigg)
\end{multlined}	
\end{aligned}	
\end{gather*}
Because the minimization is linear, the optimum is found on the boundary and thus the first term in the bracket (which is subject to minimization) will at least cancel out the last term, leaving us with
\begin{align*}
	\CE{ f (q_1,s_1) - f (q_0,s_0) }{ q_0 }
	& \leq
	K_{9}
	-
	2 H \varepsilon q_0^\intercal  \mathbf{1}_{n_q}
	\\
	& \leq
	K_{9}
	-
	K_{10} \mathbf{1}_{n_q}^\intercal q_0 
\end{align*}
which fulfills lemma \ref{lemma::foster} and therefore proves throughput optimality of our $\PNC$ policy.
\\
\rightline{$\Box$}

\section{Exemplary Applications of PNC}
\label{sec::examples}

\subsection{Dynamic Topology}

We employ a scenario as depicted in \figref{fig::ext_1_scenario}, where a mobile user equipment (UE) crosses multiple sectors, each one designated to a specific access point (AP). In each sector, the UE can only communicate with the corresponding AP. The APs are connected to a global network from which they receive packets that they are supposed to transmit to the UE.
The derived queueing network is shown in \figref{fig::ext_1_model}. We use a most simplified model to yield easily interpretable results: First, the DTMC is deterministic which allows us to fix the time behavior of the transmission success probabilities $\e{m}^j_t$ of the links. Second, we model this deterministic time behavior as binary sequences which are depicted in \figref{fig::ext_1_links}. This corresponds to the case, in which the UE travels with constant velocity along a known path and the sectors do not overlap. The UE remains in each sector for exactly 3 time slots, where it experiences perfect channel quality (guaranteed transmission success).
Further we assume that a single packet is created every second time slot at $q^1$, which represents the entire arrival to the system.
\begin{figure}[htbp]
  \centering
  \includegraphics[]{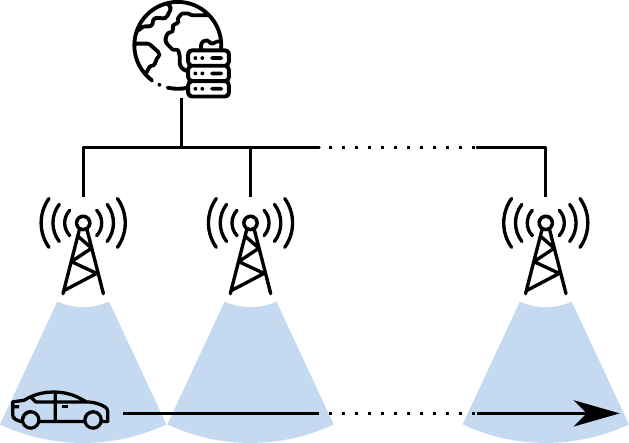}
  \caption{Extension 1 - Scenario}
  \label{fig::ext_1_scenario}
\end{figure}
\begin{figure}[htbp]
  \centering
  \includegraphics[]{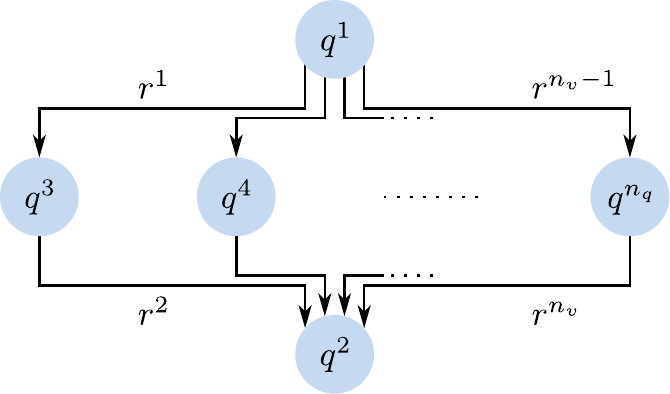}
  \caption{Example 1 - Queueing Network}
  \label{fig::ext_1_model}
\end{figure}
\begin{figure}[htbp]
  \centering
  \includegraphics[]{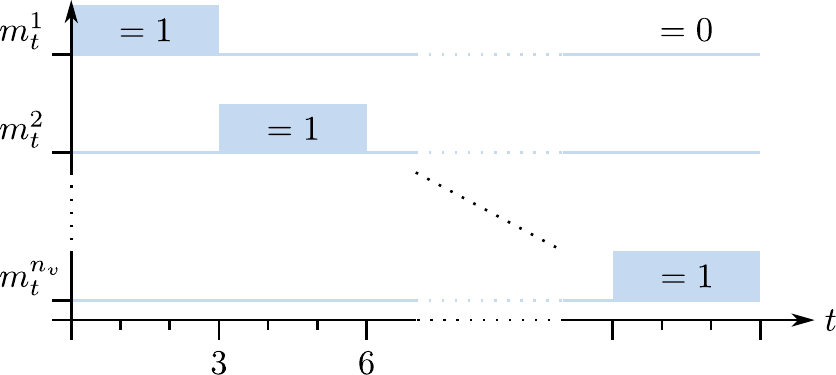}
  \caption{Example 1 - Link Probabilities}
  \label{fig::ext_1_links}
\end{figure}

The simulation results in \figref{fig::ext_1_simu_2} depict the accumulated amount of packets send (blue) and received by the UE (green and red). For visualization purposes, we averaged the resulting step functions, so that they are presented as lines. It can be observed that only around 33\% of the packets reach the UE for the conventional back pressure policy, $\MW$, (red). (Note that $\MW$ can be expressed as a special case of the $\PNC$ policy, in which the horizon is $H=1$.) The other 66\% remain at already past base stations. This high packet loss is due to $\MW$ requiring time to establish its throughput optimality. Indirectly, $\MW$ functions by using misplaced packets as an indicator for later control decisions. The presented example, however, is based on a transient event where this indicator function of misplaced packets is only of limited use.

As can be seen, using the novel $\PNC$ with horizon $H=2$, (the lowest green line) already nearly doubles the amount of packets that arrive at the UE to 60\%. For $H=5$ we reach 80\%, a significant performance boost.
\begin{figure}[htbp]
  \centering
  \includegraphics[]{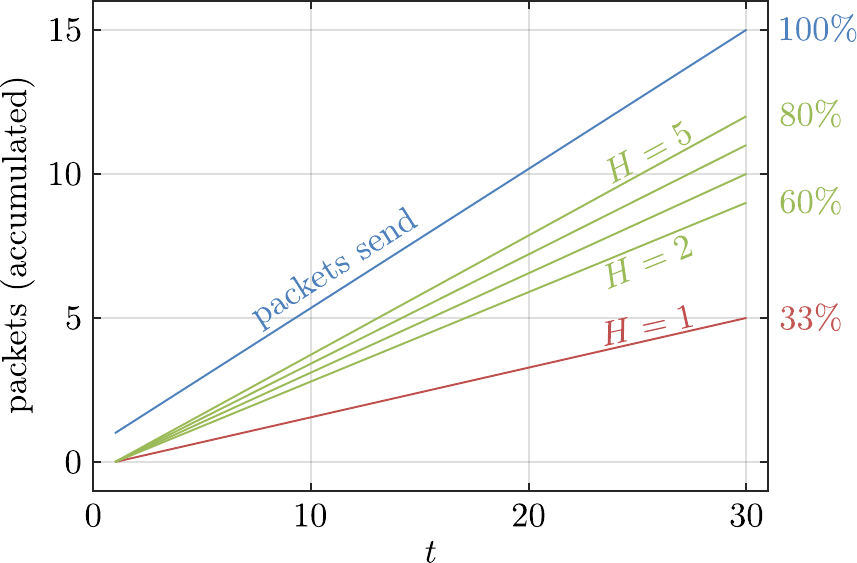}
  \caption{Example 1 - Simulation with Multiple Policies}
  \label{fig::ext_1_simu_2}
\end{figure}

\subsection{Networks with Synchronized Queues}
\label{subsec::synchronized_queues}

The following application is motivated by fact, that though $\MW$ performs poorly in networks with dynamic topology, it is still able to achieve throughput optimality in the long run, if we assume $\MW$ to be sensitive to the current state of the DTMC. And it is only fair to make this assumption, since we assume the same for our $\PNC$ policy. Hence, throughput optimality seems to be shared by both policies, if we talk about conventional networks. However, in the next example, we forgo conventional networks and introduce \textit{synchronized} queues. In networks with synchronized queues, only $\PNC$ seems to maintain its throughput optimality while $\MW$ fails, giving a strong incentive to employ the $\PNC$ policy.

Queues are \textit{synchronized} (or \textit{paired}), if they can only be served at the same time. This can be useful, if one wants to
exploit constructive interference \cite{Timotheou2016} or
model parallel processing tasks in computing \cite{Evdokimova2018} and social matchmaking \cite{Buke2015}.
While synchronized queues have been studied on their own \cite{Harrison1973} \cite{Fayolle1979} \cite{Borst2008} \cite{DeCuypere2014}, there has not been any research on how they behave in a network. Indeed, \cite{Schoeffauer2018a} presents a simple example, proving that $\MW$ loses its throughput optimality if the network contains only a single pair of synchronized queues. The reason for that can be found equation A.18 from the original proof in \cite{Tassiulas1992}, which loses its generality. In layman's terms, the original proof is based on the fact, that the evolution of the queue vector constitutes a DTMC by itself. And for conventional networks, there exists a \textit{finite} set of states (of that DTMC) which can be shown to be recurrent. This makes the entire DTMC recurrent which corresponds to throughput optimality. However, introducing synchronized queues, the finite set grows to infinite size, invalidating this correspondence. This leads to the questions, in how far back-pressure policies are suited for such networks and if there exists another policy, which guarantees throughput optimality.

To illustrate that $\PNC$ might be that policy, we refer to the example, depicted in \figref{fig::constructive_interference_set_up}. Set-up and thereof derived queueing network are shown on the left and right side, respectively. The example consists of an access point (AP) $q_1$ that can either transmit solitary (link $r^1$), or initiate synchronized transmission (link $r^3$) with a neighboring AP $q_2$. The synchronized transmission uses constructive interference and thus achieves higher throughput. However, before synchronized transmission can be initiated, the data packets have to be shared (link $r^2$), i.e. copied from $q_1$ to $q_2$.
\begin{figure}[htbp]
  \centering
  \includegraphics[]{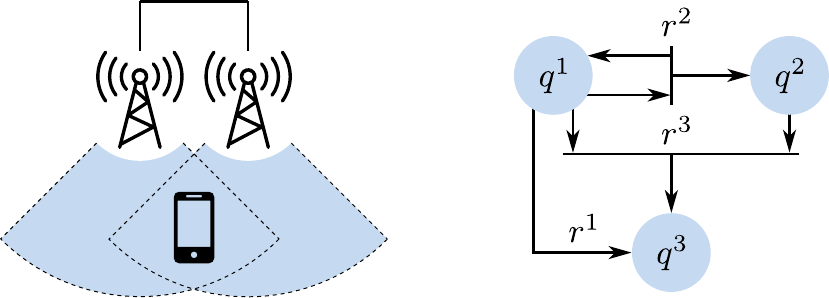}
  \caption{Example 2 - Scenario and Queueing Network}
  \label{fig::constructive_interference_set_up}
\end{figure}

For comprehensiveness, we use a constant success probability matrix $\e{M}$, i.e. we do not make use of an DTMC to select different matrices from $\mathcal{W}$. The $\e{m}^j$ (which are the diagonal elements of $\e{M}$) are chosen in such a way, that it is beneficial to copy (share) the data and then transmit together, instead of broadcasting the data directly. 
Specifically, we set $\e{m}^1 = \frac{1}{4}$ and $\e{m}^2 = \e{m}^3 = 1$ and assume all links to be disjunct. Note, that we can neglect $q^3$ in all further discussions, since it only symbolizes the destination queue.

For this simple example, it is prudent to forgo working in terms of the \textit{control vector} $v_t$ and instead use the \textit{control option} $u_t$, which we define to be $u_t = R \e{M} v_t$ (this $u_t$ is not related to the one, that was used in the definition of the $\PNC$ policy). We have $u_t \in \mathcal{U}$, which can be derived from the system description without further ado to be
\begin{equation}
    \mathcal{U} =
    \Set{
    u^0 , u^1 , u^2 , u^3
    }
    = 
    \Set{
    \begin{pmatrix} 0 \\ 0 \end{pmatrix} ,
    \begin{pmatrix} -1 \\ 0 \end{pmatrix} ,
    \begin{pmatrix} 0 \\ 4 \end{pmatrix} ,
    \begin{pmatrix} -4 \\ -4 \end{pmatrix}
    }
\end{equation}
where we scaled all elements of $\mathcal{U}$ with the factor $4$ to simplify any calculations. We have $u^1, u^2, u^3$ represent single transmission, data sharing, and joint transmission, respectively and in each time slot, the controller may only choose one of these options to influence the expected queue state via $\CE{q_{t+1}}{q_t} = q_t + u_t + a_t$.

Regarding \eqref{eq::def_to}, it is now very easy to express the set of all arrival rates $\e{a}$, for which there exists a policy that stabilizes the system. We call this set the maximum stability region $\mathcal{A}$ and have
\begin{equation}
	\begin{aligned}
	\mathcal{A} :&= \Set{
		\e{a} : \quad
		\e{a} + \sum_{u \in \mathcal{U}} \lambda^u u = -\mathbf{1} \varepsilon, \quad
		\begin{gathered}
			\varepsilon > 0
			\\
			\sum \lambda^u \leq 1
		\end{gathered}
	}
	\\
	&=
	\Set{
		\e{a} : \quad
		\e{a} + \sum_{u \in \mathcal{U}} \lambda^u u =  \mathbf{0}, \quad \hspace{4.5mm}
			\sum \lambda^u < 1
	}
	\end{aligned}
\end{equation}
Remember that throughput optimality is accomplished, if a policy can stabilize the system for all arrival rates in $\mathcal{A}$.
A graphical representation of $\mathcal{A}$ is given in green on the left side of \figref{fig::constructive_interference_stab_region}.

Now, let us assume that there is no arrival at $q^2$, i.e. $\e{a}^2 = 0$.
Using the control options $u^2$ and $u^3$ in alternating sequence (given that there are enough packets in $q^1$ to do so) would yield an efflux of $4$ packets every $2$ time slot, thus an efflux of $2$ packets per time slot. The corresponding point is shown on the right side of \figref{fig::constructive_interference_stab_region}.
It is easy to check, that no other sequence of control options can match this efflux.

However, conventional back-pressure policies like $\MW$ are not able to access the control option $u^2$, resulting in the loss of its throughput optimality in this example. The only arrival rates, that $\MW$ \textit{can} stabilize are those in the red triangle on the right side of \figref{fig::constructive_interference_stab_region}.

\begin{figure}[htbp]
  \centering
  \includegraphics[]{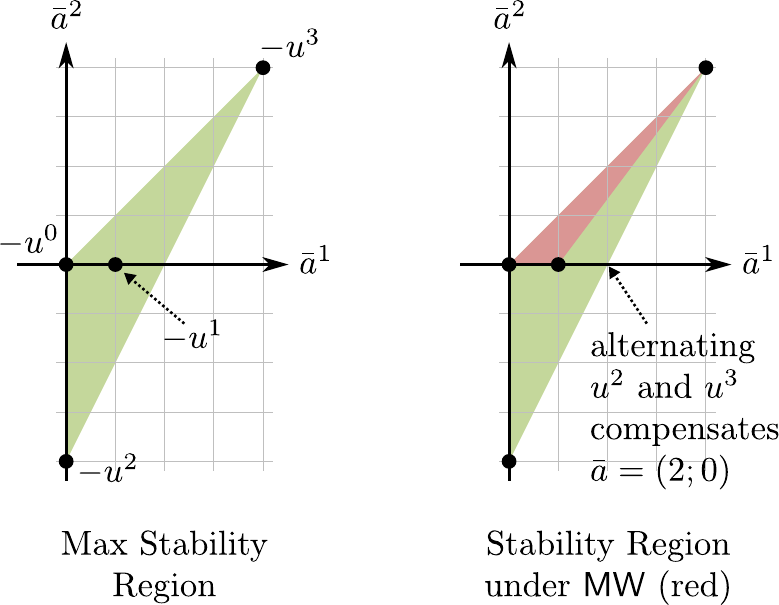}
  \vspace{0mm}
    \centering
  \caption{Example 2 - Stability Regions}
  \label{fig::constructive_interference_stab_region}
\end{figure}

In contrast, $\PNC$ is able to select the missing control option $u^2$ and simulations suggest, that it stabilizes the example for any strictly positive arrival rate $\e{a}$ from $\mathcal{A}$. To substantiate this claim we refer to \figref{fig::simu_comparison}. Here, we simulated the queue state $q^1$ over time $t$ for 3 different arrival rates $\e{a}$ under 3 different policies. The respective positions of those $\e{a}$ regarding $\mathcal{A}$ are depicted in \figref{fig::simu_points}. As for the policies, we chose $\MW$ and $\PNC$. Also, we added a third control policy, labeled $\fPNC$ for \textit{fixed} $\PNC$. This policy mimics the $\PNC$ policy, except that it uses the entire calculated control trajectory before repeating the optimization. In contrast, $\PNC$ repeats the optimization every time slot again.

As predicted, we have $\MW$ not stabilizing the blue and green arrival rates. Furthermore, it can be seen, that $\fPNC$ loses some stabilizing properties with increasing horizon as the green arrival rate cannot be stabilized with $H=3$ (this is related to the horizon not being an even number). This proves, that the MPC paradigm of repeating the optimization in every step (and thereby discarding the rest of the trajectory) is an essential part in the $\PNC$ policy.

\begin{figure}[htbp]
  \centering
  \includegraphics[]{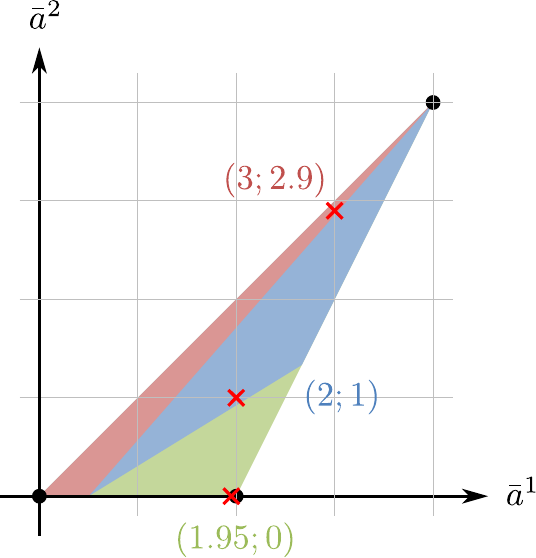}
  \caption{Example 2 - Selected arrival rates for simulation; Corresponding stability regions: (red -- $\MW$), (red+blue -- $\fPNC_{H=3}$), (red+blue+green -- $\PNC$, $\fPNC_{H=2}$)}
  \label{fig::simu_points}
\end{figure}
\begin{figure}[htbp]
  \centering
  \includegraphics[]{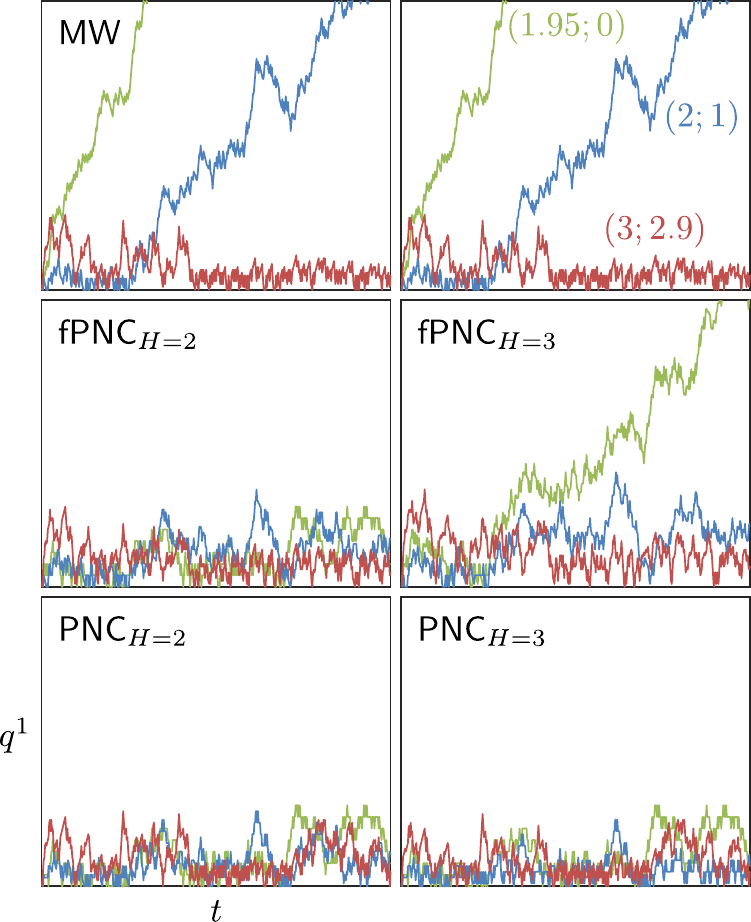}
  \caption{Example 2 - 
  Queue state (system state) $q^1$ as a function of time $t$ for various (color-coded) arrival rates}
  \label{fig::simu_comparison}
\end{figure}

\section{Conclusion}
We successfully modify a discrete-time queueing network with a JMS, i.e. with an additional DTMC that changes network parameters (even topology) on a mid- to long-term time scale. We then introduce a novel family of predictive control policies, $\PNC$, based on the paradigms of MPC, and devise a special implementation of the underlying prediction, that allows the policy to be executed in the fastest way possible. The policy is especially well suited to control the mentioned systems and outperforms conventional control approaches as is illustrated in numerical simulations. In our main contribution, we prove throughput optimality of $\PNC$. Looking ahead, we see an intriguing application in networks that consist of synchronized queues (e.g. found in parallel computing or manufacturing chains). Those networks still elude known control strategies but seem to be stabilizable under $\PNC$ policies with suitably chosen prediction horizon.

\section*{Acknowledgment}
This work is part of and thereby funded by the DFG Priority Program 1914

\bibliographystyle{ieeetr}
\bibliography{oii_2020_03_10}

\end{document}